\documentclass[reqno]{amsart}

\usepackage[hyphens]{url}
\usepackage[bookmarks]{hyperref}
\usepackage[hyphenbreaks]{breakurl}
\usepackage[style=numeric,sorting=none,backend=biber]{biblatex}
\bibliography{defcall}
\usepackage{tiz-note}

\newcommand{\app}[2]{\textsf{apply}(#1, #2)}
\newcommand{\param}[2]{\textsf{params}(#1, #2)}
\newcommand{\retTy}[2]{\textsf{typeOf}(#1, #2)}
\newcommand{\vars}[1]{\textsf{vars}(#1)}
\newcommand{\genLam}[2]{\textsf{lambda}(#1, #2)}

\newcommand{\Nat}{\Dat{Nat}}
\renewcommand{\lam}[2]{\lambda #1.~#2}

\begin{document}

\title{Elegant elaboration with function invocation}

\author{Tesla Zhang}
\address{The Pennsylvania State University}
\email{yqz5714@psu.edu}
\date\today

\begin{abstract}
We present an elegant design of the core language in a dependently-typed
lambda calculus with $\delta$-reduction and an elaboration algorithm.
\end{abstract}
\maketitle
\section{Introduction}
\label{sec:intro}

Throughout this paper, we will use \fbox{boxes} in the following two cases:
\begin{itemize}
\item to clarify the precedences of symbols when the formulae become too large. \\
e.g. \fbox{$\Gvdash \fbox{$\lam x M$}~:~\fbox{$(y:A)\to B$}
\Leftarrow \fbox{$\lam x u$}$}.
\item to distinguish type theory terms from natural language text. \\
e.g. we combine a term \fbox{$a$} with a term \fbox{$b$} to get a term \fbox{$a~b$}.
\end{itemize}

In the context of practical functional programming languages,
functions can be either \textit{primitive} (like arithmetic operations for primitive numbers,
cubical primitives~\cite[\S 3.2, \S 3.4, \S 5.1]{TR-X} such as
\textsf{transp}, \textsf{hcomp}, \textsf{Glue}, etc.)
or \textit{defined} (that have user-definable reduction rules, e.g. by using pattern matching).
The reduction of function applications are known as $\delta$-reduction~\cite{ACM-Delta}
and the to-be-reduced terms are called $\delta$-redexes
(\textit{redex} for \textit{red}ucible \textit{ex}pressions).

In general, function definitions may reduce only when \textit{fully applied},
while they are also \textit{curried}, allowing the flexible partial application.
So, from the type theoretical perspective, functions are always \textit{unary}
and function application as an operation is \textit{binary},
while in the operational semantics, function application is n-ary and $\delta$-reduction
only works when enough arguments are supplied.
We may use \textit{elaboration}, a process that transforms a user-friendly
\textit{concrete syntax} into a compact, type theoretical \textit{core language}
using type information, to deal with this inconsistency between the intuitive
concrete syntax and the optimal internal representation.

For defined functions, even when they are sometimes elaborated
as a combination of lambdas and case-trees~\cite{DepPM} or eliminators~\cite{Goguen06}
(so we can think of the function syntax as a syntactic sugar of lambdas),
the elaborated terms are related to the implementation detail of the programming language, which 
we tend to hide from the users.

To ensure the functions are fully applied, we may check whether sufficient arguments
 are supplied every time we want to reduce an application to a function 
Elaboration is also helpful here: we could choose an efficient representation of
function application in the core language and elaborate a user-friendly concrete
syntax into it. Consider a term \fbox{$\func{max}~a~b$} where $\func{max}$
is a function taking two natural numbers and returning the larger one,
we present two styles of function applications in the core language.

\begin{itemize}
\item \textit{binary application}.
The syntax definition of application is like $e::=e_1~e_2$.
The above term will be structured as \fbox{$\fbox{$\func{max}~a$}~b$}.
\item \textit{spine-normal form}.
The syntax definition of application is like $e::=e_1~\overline{e}$
\\($\overline e$ is called a \textit{spine}),
where arguments are collected as a list in application terms.
The above term will be structured as-is.
\end{itemize}

If we choose the binary representation,
checking the number of supplied arguments requires a traversal of the terms,
which is an $O(n)$ process where $n$ is the number of parameters of the applied function
(accessing the applied function from an application term is also $O(n)$).

This is not a problem in spine-normal form representation, where we only need to check
if the size of the arguments is larger than the number of parameters during $\delta$-reduction.
Size comparison is an $O(1)$ process for memory-efficient arrays.
However, inserting extra arguments to the spine is a list mutation operation,
which requires a list reconstruction in a purely functional setting.
The list reconstruction creates a new list with a larger size and copies the old list
with the new argument appended to the end, which is again an $O(n)$ process.

Similar problems also exist for \textit{indexed types}~\cite{SIT},
where we need the types to be fully applied to determine the available constructors.

\subsection{Contributions}
\label{sub:contrib}
We discuss an even more elegant design of application term representation where
neither traversal of terms during reduction nor mutation of the argument lists
during application are needed.
We will use \textit{both} binary application \textit{and} spine-normal form
to guarantee that the number of supplied arguments always fits the requirement,
and we could always assume function applications to have sufficient arguments supplied.

We present the syntax (in~\cref{sec:syntax}) and a bidirectional-style elaboration
algorithm (in~\cref{sec:elab}) which can be adapted to the implementation of
any programming language with $\delta$-reduction.

\section{Syntax}
\label{sec:syntax}
We will assume the existence of well-typed \textit{function definitions}
(will be further discussed in~\cref{sec:elab}),
and focus on lambda calculus terms.
Inductive types and pattern matching are also assumed and omitted.

\subsection{Core language}

The syntax is defined in~\cref{fig:term}.

Core terms have several assumed properties: well-typedness, well-scopedness,
and the functions are \textit{exactly-applied} --
the number of arguments matches exactly the number of parameters.

\begin{figure}[ht!]
\begin{align*}
  x,y     ::= & && \text{variable names} \\[-0.3em]
  A,B,u,v ::= & \quad \func f~\overline u && \text{exactly-applied function} \\[-0.3em]
    \mid & \quad x && \text{reference} \\[-0.3em]
    \mid & \quad u~v && \text{binary application} \\[-0.3em]
    \mid & \quad (x:A)\to B && \text{$\Pi$-type} \\[-0.3em]
    \mid & \quad \lam x u && \text{lambda abstraction} \\[-0.3em]
  \sigma ::= & \quad \overline{u/x} && \text{substitution object}
\end{align*}
\caption{Syntax of terms}
\label{fig:term}
\end{figure}

We will assume the substitution operation on core terms,
written as \fbox{$u~[\sigma]$}.

\begin{example}
\label{ex:nest}
For a function $\func f$ who has two parameters
and returns a function with two parameters,
the application $\func f~u_1~u_2~u_3~u_4$
is structured as \fbox{$\fbox{$\fbox{$\func f~u_1~u_2$}~u_3$}~u_4$}.
The innermost subterm $\func f~u_1~u_2$ is an exactly-applied function
and the outer terms are binary applications.  
\end{example}


\subsection{Concrete syntax}
The concrete syntax is defined in~\cref{fig:expr}.
We assume the names to be resolved in concrete syntax,
so we can assume the terms to be well-scoped and
distinguish local references and function references.

The symbol for local references and functions are overloaded
because they make no difference between concrete and core.

\begin{figure}[h!]
\begin{align*}
  M,N ::= & \quad \func f && \text{function reference} \\[-0.3em]
    \mid & \quad x && \text{local reference} \\[-0.3em]
    \mid & \quad M~N && \text{application} \\[-0.3em]
    \mid & \quad \lam x M && \text{lambda abstraction}
\end{align*}
\caption{Concrete syntax tree}
\label{fig:expr}
\end{figure}

In the concrete syntax, applications are always binary.

We do not have $\Pi$-types in concrete syntax because
they are unrelated to application syntax and their type-checking
is relevant to the \textit{universe type} which is quite complicated.

\section{Elaboration}
\label{sec:elab}
In this section, we describe the process that type-checks concrete terms against
core terms and translates them into well-typed, exactly-applied core terms.

We define the following operation $\app u v$ to eliminate
obviously reducible core terms:

\begin{align*}
\app{\lam x u}v & ::= u[v/x] \\
\app u v & ::= u~v
\end{align*}

We will use a bidirectional elaboration algorithm~\cite{BidirTyck,MiniTT}.
It uses the following typing judgments, parameterized by a context $\Gamma$:

\begin{itemize}
\item $\Gvdash M:A \Leftarrow u$,
 normally called a \textit{checking} or an \textit{inherit} rule.
\item $\Gvdash M \Rightarrow u:A$,
 normally called an \textit{inferring} or a \textit{synthesis} rule.
\item $x:A \in \Gamma$, that the context $\Gamma$ contains a local binding $x:A$.
\end{itemize}

Checking judgments are for introduction rules,
while synthesis judgments are for the formation and elimination rules.
The direction of the arrows is inspired from~\cite[\S 6.6.1]{MiniTT}.
The arrows ($\Leftarrow$ and $\Rightarrow$) separate the inputs and outputs
of the corresponding elaboration procedure:
checking judgments take as input a term in concrete form and its expected type in core form,
and produce as output an elaborated version of the term in core form;
synthesis judgments take as input a term in concrete form and emit its elaborated form and the
inferred type as output.

The context $\Gamma$ contains a list of local bindings
(pairs of local variables and types, written as $x:A$)
and a list of functions.

A function is defined with a name and a signature
(and a body that we do not care about in this paper),
where the signature tells us about its parameters and the return type.

\subsection{Abstraction and application}

First of all, we have the basic elaboration rules for
type conversion and local references: 

\begin{mathpar}
\inferrule{\Gvdash M \Rightarrow u:A \\ \Gvdash A =_{\beta\eta} B}
{\Gvdash M:B \Leftarrow u}~\textsc{Conv} \and
\inferrule{x:A \in \Gamma}{\Gvdash x \Rightarrow x:A}~\textsc{Var}
\end{mathpar}

The rules for lambda abstraction and application are
quite straightforward:

\begin{mathpar}
\inferrule{\Gamma,x:A \vdash M:B[x/y] \Leftarrow u}
{\Gvdash \fbox{$\lam x M$}~:~\fbox{$(y:A)\to B$}
\Leftarrow \fbox{$\lam x u$}}~\textsc{Lam}
\and
\inferrule{\Gvdash M \Rightarrow u:\fbox{$(x:A)\to B$} \\
\Gvdash N:A \Leftarrow v}
{\Gvdash M~N
\Rightarrow \app u v:B[v/x]}~\textsc{App}
\end{mathpar}

\subsection{Functions}

Before continuing to function references,
we need to define the notation for function signature.
We define \textit{parameters} to be a list of local bindings:

$$
\Delta ::= \overline{x:A}
$$

Then, we assume two operations for every function symbol $\func f$:

\begin{itemize}
\item $\param\Gamma{\func f}$ that returns a $\Delta$
 that represents the parameters of $\func f$ in $\Gamma$.
\item $\retTy\Gamma{\func f}$ that returns a term \fbox{$u$}
 that represents the type (combining parameters and the return type) of $\func f$ in $\Gamma$.
\end{itemize}

We will need a few more operations before introducing the
elaboration rules for functions.

$\vars\Delta$: extracts the variables (as a list of terms) from $\Delta$:

\begin{align*}
\vars\emptyset &::=\emptyset \\
\vars{x:A,\Delta} &::= x,\vars\Delta \\
\end{align*}

$\genLam u\Delta$: generates a lambda abstraction by induction on $\Delta$:

\begin{align*}
\genLam u\emptyset &::=u \\
\genLam u{\fbox{$x:A,\Delta$}} &::= \lam x {\genLam u\Delta} \\
\end{align*}

With these two operations, we can define the elaboration
rule for functions. Let $\Delta = \param\Gamma{\func f}$,
we have the synthesis rule for function references as in~\cref{eqn:fun}.

\begin{figure}[h!]
\begin{mathpar}
\inferrule{\func f \in \Gamma}
{\Gvdash \func f \Rightarrow \genLam{\fbox{$\func f~\vars\Delta$}}
\Delta : \retTy\Gamma{\func f}}
\end{mathpar}
\caption{Elaboration of functions}
\label{eqn:fun}
\end{figure}

Here are some examples.
Consider the $\func{max}$ function discussed in~\cref{sec:intro},
we have the following facts:

\begin{equation*}
\begin{aligned}
\param\Gamma{\func{max}} &::= x:\Nat,y:\Nat \\
\retTy\Gamma{\func{max}} &::= (x:\Nat)\to(y:\Nat)\to\Nat
\end{aligned}
\end{equation*}

\begin{example}
\label{ex:max}
Concrete term \fbox{$\func{max}$}.
By the rule in~\cref{eqn:fun},
\begin{equation*}
\begin{aligned}
 & \genLam{\func{max}~\vars\Delta}\Delta & \\
 &= \genLam{\fbox{$\func{max}~\vars{x:\Nat,y:\Nat}$}}{\fbox{$x:\Nat,y:\Nat$}}
  & \text{Expand}~\Delta \\
 &= \genLam{\func{max}~x~y}{\fbox{$x:\Nat,y:\Nat$}} & \text{Expand}~\textsf{vars} \\
 &= \lam x{\lam y{\func{max}~x~y}} & \text{Expand}~\textsf{lambda}
\end{aligned}
\end{equation*}
Observe that $\func{max}$ is exactly-applied.
The result type is $\retTy\Gamma{\func{max}}$,
which equals $(x:\Nat)\to(y:\Nat)\to\Nat$.
\end{example}

\begin{example}
Concrete term \fbox{$\func{max}~M$}, assuming $\Gvdash M:\Nat \Leftarrow u$.
We already know the elaboration result of $\func{max}$ in~\cref{ex:max}.
By the \textsc{App} rule, since
$\app{\fbox{$\lam x{\lam y{\func{max}~x~y}}$}}u \mapsto \fbox{$\lam y{\func{max}~u~y}$}$,
the elaborated version of $\func{max}~M$ is $\lam y{\func{max}~u~y}$,
typed $(y:\Nat)\to\Nat$.
Observe that $\func{max}$ is still exactly-applied.
\end{example}

\begin{thm}
Functions are always exactly-applied in the core language,
as promised in~\cref{sub:contrib}.
\end{thm}
\begin{proof}
We only generate exactly-applied functions in~\cref{eqn:fun},
and we do not take arguments away or insert new arguments to
function application terms in any other rules or operations.
\end{proof}

\section{Conclusion}
We have discussed an elegant core language design with $\delta$-redexes
with an elaboration algorithm.
With this design, wherever in the compiler,
we can assume any given $\delta$-redexes to be exactly-applied.
Appending extra arguments to an application term results in a binary
application term to a non-function (as in~\cref{ex:nest}).

\subsection{Implementation}
The discussed core language design is implemented in two proof assistants:

\begin{itemize}
\item Arend~\cite{Arend}.
There is an abstract class \textsf{DefCallExpression}
that generalizes all sorts of
\textit{definition invocations}, including functions, data types, constructors, etc.
These expressions are always exactly-applied.
The source code is hosted on GitHub~\footnote{
  See \url{https://github.com/JetBrains/Arend}}.
\item Aya~\cite{Aya}.
Similar to Arend, there is an interface \textsf{CallTerm}
in Aya. The source code is also hosted on GitHub~\footnote{
  See \url{https://github.com/aya-prover/aya-dev}}.
\end{itemize}

In the implementations, there is one extra complication: invocations to constructors of
inductive types have access to the parameters of the inductive type.

\subsection{Related work}
The notion of $\delta$-reduction was discussed in~\cite{ACM-Delta},
but the $\delta$-redexes are represented in binary application form.
Exactly-applied $\delta$-redexes are discussed in~\cite{DeltaRed, DepPM},
but they did not discuss elaboration.
The idea of separating spine-normal function application and binary application
also appeared in an early work on LISP~\cite[\S 6]{LISP}, where spines referred to as
\textit{rails} and binary application referred to as \textit{pairs},
but they also did not discuss elaboration.

In the elaboration of Lean 2~\cite{Lean2},
functions are transformed into lambdas and \textit{recursors}
(and they refer to the corresponding redexes as
$\beta$-redexes and $\iota$-redexs~\cite[\S 3.3]{Lean2}, respectively).
This design is not friendly for primitive functions that only work when fully applied.

\subsection{Acknowledgments}
We would like to thank Marisa Kirisame, Ende Jin, Qiantan Hong, and Zenghao Gao
for their comments and suggestions on the draft versions of this paper.

\printbibliography
\end{document}